%% file: main.tex
\documentclass[letterpaper, 10 pt, conference]{ieeeconf} 
\IEEEoverridecommandlockouts                              

\usepackage{graphicx}
\usepackage{subcaption}
\usepackage{etoolbox}
\usepackage{amsfonts}
\usepackage{newpxtext,newpxmath}
\usepackage{scalerel}

\usepackage[htt]{hyphenat}
\usepackage{siunitx}
\usepackage{pifont}
\usepackage{graphicx}
\usepackage{amsmath,textcomp}
\usepackage{tikz}

\allowdisplaybreaks

\usepackage{booktabs}
\usepackage{url}
\usepackage{algpseudocode}
\usepackage{paralist}
\usepackage{color}
\usepackage{stmaryrd}
\usepackage{epsfig} 
\usepackage{bm} 
\usepackage{listings,lstautogobble}
\usepackage{mathtools}
\usepackage{mathrsfs}
\usepackage{xspace}
\usepackage{verbatim}
\usepackage{epstopdf}
\usepackage{listings}
\usepackage{parcolumns}
\usepackage{color}
\usepackage{xspace}
\usepackage{verbatim}
\usepackage[utf8]{inputenc}
\usepackage{calrsfs}
\usepackage{rotating,multirow}

\usepackage[ruled,vlined,linesnumbered]{algorithm2e}
\SetKwInput{KwInput}{Input}
\SetKwInput{KwOutput}{Output}
\SetKwInOut{KwParam}{Parameter}

\usepackage[T1]{fontenc}
\usepackage[utf8]{inputenc}

\newtheorem{example}{Example}
\newtheorem{definition}{Definition}

\newtheorem{theorem}{Theorem}
\newtheorem{remark}{Remark}
 \newtheorem{lemma}{Lemma}

  \newtheorem{corollary}{Corollary}

  \renewcommand{\footnotesize}{\fontsize{8pt}{11pt}\selectfont}

\DeclareMathAlphabet{\mathcal}{OMS}{cmsy}{m}{n}

\usepackage{hyperref}
\hypersetup{
    colorlinks=true,
    linkcolor=black,
    urlcolor=gray,
}
\usepackage{balance}

\newcommand{\bigzono}[1]{\Big\langle #1 \Big\rangle}
\newcommand{\zono}[1]{\langle #1 \rangle}

\newcommand*\xor{\oplus}
\newcommand*\xnor{\odot}
\newcommand*\xorsum[2]{\overset{#2}{\underset{#1}{\xor}}}

\makeatletter
\DeclareRobustCommand{\nand}{\mathbin{\mathpalette\n@and@or\land}}
\DeclareRobustCommand{\nor}{\mathbin{\mathpalette\n@and@or\lor}}

\newcommand{\n@and@or}[2]{%
  \vphantom{#2}%
  \ooalign{$\m@th#1#2$\cr\hidewidth$\m@th#1\sim$\hidewidth\cr}%
}
\makeatother

\begin{document}
\title{\LARGE \bf Logical Zonotopes: A Set Representation for the \\ Formal Verification of Boolean Functions}

\author{Amr Alanwar$^{1,2}$, Frank J. Jiang$^{3}$, Samy Amin$^{2}$, and Karl H. Johansson$^{3}$
\vspace*{-5mm}
\thanks{This work is supported by the Knut and Alice Wallenberg Foundation, the Swedish Strategic Research Foundation, the Swedish Research Council, and the  Wallenberg AI, Autonomous Systems and Software Program (WASP) funded by the Knut and Alice Wallenberg Foundation.}
\thanks{$^{1}$School of Computation, Information and Technology, Technical University of Munich. {\tt\small alanwar@tum.de}.}
\thanks{$^{2}$School of Computer Science and Engineering, Constructor University. {\tt\small \{aalanwar,samin\}@constructor.university}.}%
\thanks{$^{3}$School of Electrical Engineering and Computer Science, KTH Royal Institute of Technology. Authors are affiliated with Digital Futures.  {\tt\small\{frankji, kallej\}@kth.se.}}%
}




\maketitle
\input{Sections/Aabs}

\input{Sections/Bintro}

\input{Sections/Cprelim}
\input{Sections/DmainOperations}
\input{Sections/Dreduce}
\input{Sections/Dreach}

\input{Sections/Eeval}
\input{Sections/Fcon}


\bibliographystyle{IEEEtran}
{\small
\bibliography{IEEEabrv,ref}
}

\end{document}

%% file: Sections/Aabs.tex
\begin{abstract}
A logical zonotope, which is a new set representation for binary vectors, is introduced in this paper. A logical zonotope is constructed by XOR-ing a binary vector with a combination of other binary vectors called generators. Such a zonotope can represent up to $2^\gamma$ binary vectors using only $\gamma$ generators. It is shown that logical operations over sets of binary vectors can be performed on the zonotopes' generators and, thus, significantly reduce the computational complexity of various logical operations (e.g., XOR, NAND, AND, OR, and semi-tensor products). Similar to traditional zonotopes' role in the formal verification of dynamical systems over real vector spaces, logical zonotopes can efficiently analyze discrete dynamical systems defined over binary vector spaces. We illustrate the approach and its ability to reduce the computational complexity in two use cases: (1) encryption key discovery of a linear feedback shift register and (2) safety verification of a road traffic intersection protocol.
\end{abstract}

%% file: Sections/Bintro.tex
\section{Introduction}
For several decades, logical systems have been used to model complex behaviors in numerous applications. By modeling a system as a collection of logical functions operating in a binary vector space, we can design models that consist of relatively simple dynamics but still capture a complex system's behavior at a sufficient level of abstraction. Some popular approaches to modeling logical systems are finite automatons, Petri nets, or Boolean Networks (BNs). 

%
%

An important form of analysis for logical systems is reachability analysis. Reachability analysis allows us to formally verify the behavior of logical systems and provide guarantees that, for example, the system will not enter into undesired states. One of the primary challenges of reachability analysis is the need to exhaustively explore the system's state space, which grows exponentially with the number of state variables. To avoid exponential computational complexity, many reachability analysis algorithms are based on a representation called Binary Decision Diagram (BDD). Given a proper variable ordering, BDDs can evaluate Boolean functions with linear complexity in the number of variables~\cite{conf:BDDthesis}.
While BDDs play a crucial role in verification, they have well-known drawbacks, such as requiring an externally supplied variable ordering~\cite{conf:npcomplete, conf:effreachBDD}.
Outside of BDDs, there are also approaches to reachability analysis for logical systems modeled as BNs, or Boolean Control Networks (BCNs) for systems with control inputs, that rely on the semi-tensor product~\cite{7454743}. However, due to being point-wise and scaling limitations of semi-tensor products, BCN-based approaches become intractable for high-dimensional logical systems~\cite{leifeld2019overview}.
In this work, we propose a novel zonotope representation that significantly reduces the exponential computational complexity of reachability analysis. 

\begin{figure}
    \centering
    \includegraphics[scale=0.28]{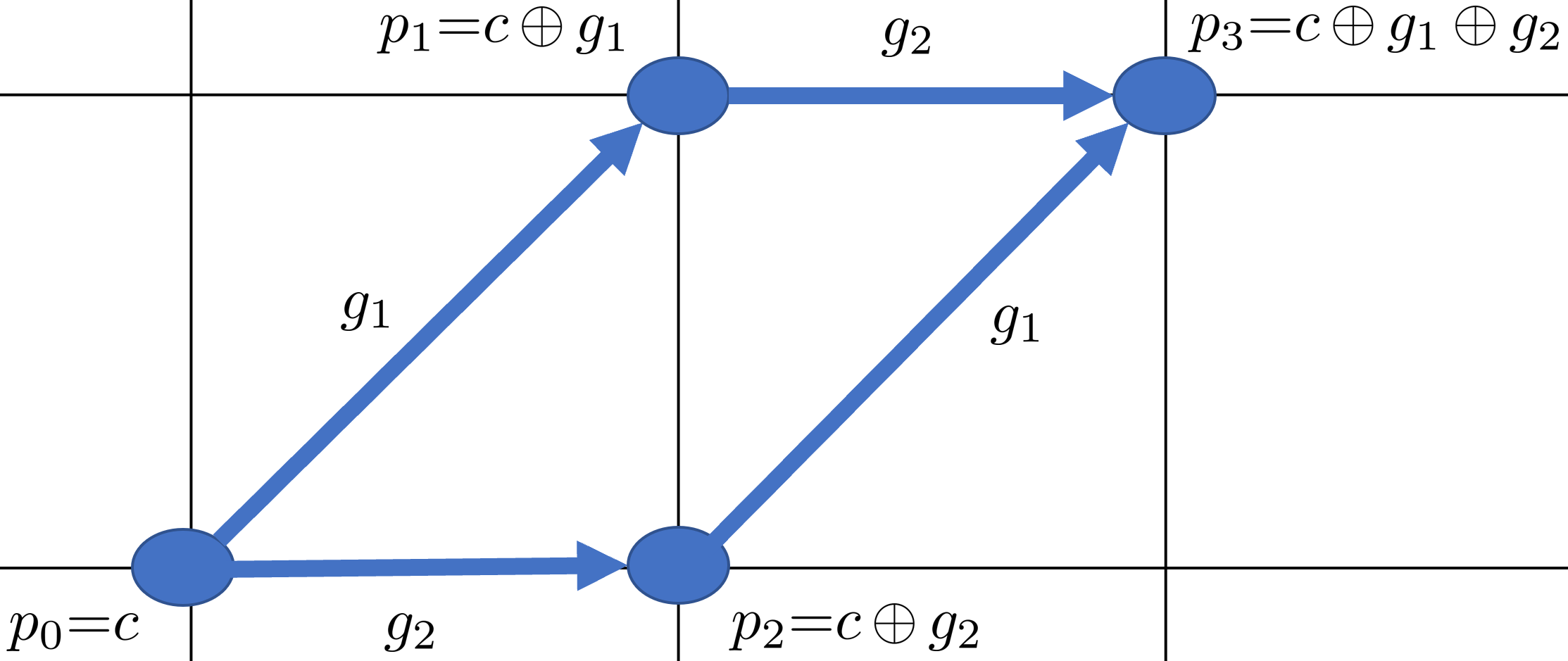}
    \caption{Representing four points $p_0,\dots,p_3$ by considering all combinations of turning off and on the two generators $g_{1}$ and $g_{2}$ of a logical zonotope.}
    \label{fig:LogicalZono}
    \vspace{-6mm}
\end{figure}

In real vector spaces, zonotopes already play an important role in the reachability analysis of dynamical systems~\cite{girard, conf:thesisalthoff}. Classical zonotopes are constructed by taking the Minkowski sum of a real vector center and a combination of real vector generators. Through this construction, a set of infinite real vectors can be represented by a finite number of generators. Then, by leveraging the fact that the Minkowski sum of two classical zonotopes can be computed by combining their respective generators, researchers have formulated computationally efficient approaches to reachability analysis~\cite{conf:thesisalthoff}. In this work, we take inspiration from classical zonotopes and formulate logical zonotopes. Similarly, logical zonotopes are constructed by XOR-ing a binary vector center and a combination of binary vector generators. In binary vector spaces, logical zonotopes are able to represent up to $2^\gamma$ binary vectors using only $\gamma$ generators, as illustrated in Figure~\ref{fig:LogicalZono} with $\gamma=2$. Moreover, we show that any logical operation on the generators of the logical zonotopes is either equivalent to or over-approximates the explicit application of the logical operation to each member of the represented set. Based on these results, we formulate our logical zonotope-based reachability analysis for logical systems.

Explicitly, the contributions of this work is summarized by the following:
(1) we present our formulation of logical zonotopes,
(2) we detail the application of logical operations and forward reachability analysis to logical zonotopes,
(3) we illustrate the use of logical zonotopes in two different applications.
To recreate our results, readers can use our publicly available logical zonotope library\footnotemark.

\footnotetext{\href{https://github.com/aalanwar/Logical-Zonotope}{https://github.com/aalanwar/Logical-Zonotope}}

The remainder of the paper is organized as follows. In Section~\ref{sec:prelim}, we introduce the notation and preliminary definitions. In Section~\ref{sec:approach}, we formulate logical zonotopes and overview the supported operations. Then, in Section~\ref{sec:eval}, we illustrate the applications of logical zonotopes. Finally, in Section~\ref{sec:con}, we conclude the work.

%% file: Sections/Cprelim.tex
\section{Notation and Preliminaries}\label{sec:prelim}



In this section, we introduce details about the notation used throughout this work and preliminary definitions for logical systems and reachability analysis.

\subsection{Notation}
 The set of natural and real numbers are denoted by $\mathbb{N}$ and $\mathbb{R}$, respectively. 
 We denote the binary set $\{0,1\}$ by $\mathbb{B}$. The XOR, NOT, OR, and AND operations are denoted by $\xor,\neg,\lor$, and $\land$, respectively. Throughout the rest of the work, with a slight abuse of notation, we omit the $\land$ from $a \land b$ and write $a \, b$ instead. The NAND, NOR, and XNOR are denoted by $\nand,\nor$, and $\xnor$, respectively. Later, we use the same notation for both the classical and Minkowski logical operators, as it will be clear when the operation is taken between sets or individual vectors. Like the classical AND operator, we will also omit the Minkowski AND to simplify the presentation. Matrices are denoted by uppercase letters, e.g., $G \in \mathbb{B}^{n \times k}$, and sets by uppercase calligraphic letters, e.g., $\mathcal{Z} \subset \mathbb{B}^{n}$. Vectors and scalars are denoted by lowercase letters, e.g., $b \in \mathbb{B}^{n }$ with elements $b_{1:n}$. The identity matrix of size $n \times n$ is denoted $I_n$. We denote the Kronecker product by $\otimes$. The $x \in \mathbb{B}^{n}$ is an $n \times 1$ binary vector. 

\subsection{Preliminaries}
For this work, we consider a system with a logical function $f: \mathbb{B}^{n_x} \times \mathbb{B}^{n_u} \rightarrow \mathbb{B}^{n_x}$:
\begin{align}
    x(k+1) = f\big(x(k),u(k)\big)
\label{eq:feq}
\end{align}
where $x(k) \in \mathbb{B}^{n_x}$ is the state and $u(k) \in \mathbb{B}^{n_u}$ is the control input. The logical function $f$ can consist of any combination of $\xor,\neg,\lor,\nand,\nor,\xnor,$ and $\land$. We will represent sets of states and inputs for~\eqref{eq:feq} using logical zonotopes. As will be shown, logical zonotopes are constructed using an Minkowski XOR operation, which we define as follows.

\begin{definition}\label{def:semitensor}\textbf{(Minkowski XOR)} \
Given two sets $\mathcal{L}_1$ and $\mathcal{L}_2$ of binary vectors, the Minkowski XOR is defined between every two points in the two sets as
\begin{align}
      \mathcal{L}_1 \xor \mathcal{L}_2 &= \{z_1 \xor z_2| z_1\in \mathcal{L}_1, z_2 \in \mathcal{L}_2 \}. \label{eq:xor} 
\end{align}
\end{definition}
Similarly, we define the Minkowski NOT, OR and AND operations as follows.
\begin{align}
\neg \mathcal{L}_1  	&= \{ \neg z_1| z_1\in \mathcal{L}_1 \},  \\
\mathcal{L}_1\lor  \mathcal{L}_2 	&= \{z_1 \lor z_2| z_1\in \mathcal{L}_1, z_2 \in \mathcal{L}_2 \}, \\
\mathcal{L}_1  \mathcal{L}_2 	&= \{z_1  z_2| z_1\in \mathcal{L}_1, z_2 \in \mathcal{L}_2 \}.\label{eq:and}
\end{align}

We aim to show how logical zonotopes can be used to compute the forward reachable sets of systems defined by~\eqref{eq:feq}. We define the reachable sets of system~\eqref{eq:feq} by the following definition.

\begin{definition}\label{def:exactreachF}(\textbf{Exact Reachable Set}) \
Given a set of initial states $\mathcal{X}_0 \subset \mathbb{B}^{n_x}$ and a set of possible inputs $\,\mathcal{U}_k \subset \mathbb{B}^{n_u}$, the exact reachable set $\mathcal{R}_{N}$ of \eqref{eq:feq} after $N$ steps is
\begin{align*}
    \mathcal{R}_{N} = \big\{ &x(N) \in \mathbb{B}^{n_x} \; \big| \; \forall k \in \{0,...,N-1\}: \\
        & x(k+1) = f\big(x(k),u(k)\big), 
        \; x(0) \in \mathcal{X}_0, u(k) \in \mathcal{U}_k \big\}.
\end{align*}
\end{definition}



Another commonly used operator for BCNs is the semi-tensor product~\cite{conf:tensorproductsurvey}. Since semi-tensor products are useful in many applications, we have extended the classical definition to logical zonotopes. The classical definition for semi-tensor products is as follows.

\begin{definition}\label{def:semitensor}\textbf{(Semi-Tensor Product \cite{conf:semitensor})} \
Given two matrices $M \in \mathbb{B}^{m \times n}$ and $N  \in \mathbb{B}^{p \times q}$, the semi-tensor
product, denoted by $\ltimes$, is computed as:
\begin{align}
    M \ltimes N = (M \otimes I_{s_1})( N \otimes I_{s_2}),
\end{align}
where $s_1 =s/n$, and $s_2 =s/p$, with $s$ being the least common multiple of $n$ and $p$.
\end{definition}

%% file: Sections/DmainOperations.tex
\section{Logical Zonotopes}\label{sec:approach}
In this section, we present logical zonotopes and overview several different aspects of their use. We start by defining the set representation of logical zonotopes. Then, we go through the application of Minkowski XOR, NOT, XNOR, AND, NAND, OR, and NOR on logical zonotopes. 
Using these results, we show that when using logical zonotopes for reachability analysis on~\eqref{eq:feq}, we are able to compute reachable sets that over-approximate the exact reachable sets. Finally, we present an algorithm for reducing the number of generators of a logical zonotope.

\subsection{Set Representation}
Inspired by the classical zonotopic set representation which is defined in real vector space \cite{conf:zono1998}, we propose logical zonotopes as a set representation for binary vectors. We define logical zonotopes as follows.
\begin{definition}(\textbf{Logical Zonotope}) \label{df:zono}
Given a point $c {\in} \mathbb{B}^{n}$ and $\gamma {\in} \mathbb{N}$ generator vectors in a generator matrix $G{=}\begin{bmatrix} g_{1},  {\dots} ,g_{\gamma}\end{bmatrix}$ $\in \mathbb{B}^{n \times \gamma}$, a logical zonotope is defined as
\begin{align*}
\mathcal{L} = \Big\{ x \in \mathbb{B}^n \; \Big| \; x = c \xorsum{i=1}{\gamma}  g_{i} \beta_{i}, \, \beta_{i} \in \{0,1\} \Big\} \, .
\end{align*}
We use the shorthand notation $\mathcal{L} = \zono{c,G}$ for a logical zonotope. 
\end{definition}


\begin{example}
Consider a logical zonotope 
$$\mathcal{L}=\bigzono{\begin{bmatrix} 0 \\1 \end{bmatrix},\begin{bmatrix} 1 & 1 \\0 &1 \end{bmatrix}}.$$ With two generators, It represents the following four points:
\begin{align*}
    \begin{bmatrix}
    0 \\ 1
    \end{bmatrix},\begin{bmatrix}
    1 \\ 0
    \end{bmatrix},\begin{bmatrix}
    1 \\ 1
    \end{bmatrix},\begin{bmatrix}
    0 \\ 0
    \end{bmatrix}
\end{align*}
by iterating over all possible binary values of $\beta \in \{00,01,10,11\}$.
\end{example}

\begin{remark}
Logical zonotopes are defined over $\mathbb{B}^n$ and are different from zonotopes \cite{conf:zono1998}, constrained zonotopes \cite{conf:constrainedzono}, and hybrid zonotopes~\cite{conf:hybridzono} which are all defined over real vector space $\mathbb{R}^n$. Specifically, logical zonotopes are functional sets with Boolean symbols~\cite{Combastel2022}.
\end{remark}

Logical zonotopes $\mathcal{L}$ can enclose up to $2^{\gamma}$ binary vectors with $\gamma$ generators. In the following section, we will show that due to their construction, we can apply logical operations to a set of up to $2^{\gamma}$ binary vectors with a reduced computational complexity. 

\subsection{Minkowski Logical Operations}
Given two sets of binary vectors, we often need to perform logical operations between the members of the two sets. In order to perform these logical operations efficiently, we define new logical operations that only operate on the generators of logical zonotopes instead of the members contained within the zonotopes. We will go through each logical operation and show that when applied to logical zonotopes, they either yield exact solutions or over-approximations. 

\subsubsection{Minkowski XOR ($\xor$)} \,
Given the nature of the logical zonotope construction, we are able to compute the Minkowski XOR exactly and show that logical zonotopes are closed under Minkowski XOR as follows.

\begin{lemma}
\label{lem:xor}
Given two logical zonotopes $\mathcal{L}_1=\langle c_{1},G_{1} \rangle$ and $\mathcal{L}_2=\langle c_{2},G_{2} \rangle$, the Minkowski XOR is computed exactly as: 
\begin{align}
      \mathcal{L}_1 \xor \mathcal{L}_2  &= \Big\langle c_{1} \xor c_{2}, \begin{bmatrix} G_{1} , G_{2} \end{bmatrix} \Big\rangle.
     \label{eq:xormink}
\end{align}
\end{lemma}
\begin{proof}
Let's denote the right hand side of \eqref{eq:xormink} by $\mathcal{L}_{\xor}$. We aim to prove that $\mathcal{L}_1 \xor \mathcal{L}_2 \subseteq \mathcal{L}_{\xor}$ and $\mathcal{L}_{\xor} \subseteq \mathcal{L}_1 \xor \mathcal{L}_2$. Choose any $z_1 \in \mathcal{L}_1$ and $z_2 \in \mathcal{L}_2$ 
\begin{align*}
 \exists \hat{\beta}_{1} &: z_1 = c_{1} \xorsum{i=1}{\gamma_{1}}  g_{1,i} \hat{\beta}_{1,i} \, ,\\
  \exists \hat{\beta}_{2} &: z_2 = c_{2} \xorsum{i=1}{\gamma_{2}}  g_{2,i} \hat{\beta}_{2,i}\, .
 \end{align*}
 Let $\hat{\beta}_{\xor,1:\gamma_{\xor}} {=} \begin{bmatrix} \hat{\beta}_{1,1:\gamma_{1}}\,, \hat{\beta}_{2,1:\gamma_{2}} \end{bmatrix}$ with $\gamma_{\xor} {=} \gamma_{1} {+} \gamma_{2}$. Given that XOR is an associative and commutative gate, we have the following:
\begin{align*}
z_1 \xor z_2 &= c_{1} \xorsum{i=1}{\gamma_{1}}  g_{1,i} \hat{\beta}_{1,i} \xor  c_{2} \xorsum{i=1}{ \gamma_{2}}  g_{2,i} \hat{\beta}_{2,i}\\
& = c_{\xor} \xorsum{i=1}{\gamma_{1}+\gamma_{2}} g_{\xor,i} \hat{\beta}_{\xor,i}\, ,
\end{align*}
where $c_{\xor} = c_{1} \xor c_{2}$ and  $G_{\xor} = \begin{bmatrix} G_{1}\,,\, G_{2} \end{bmatrix}$ with $G_{\xor}{=}\big[ g_{\xor,1},$  ${\dots} ,g_{\xor,\gamma_{\xor}}\big]$. Thus, $z_1 \xor z_2 \in \mathcal{L}_{\xor}$ and therefore $\mathcal{L}_1 \xor \mathcal{L}_2 \subseteq \mathcal{L}_{\xor}$. Conversely, let $z_{\xor} \in \mathcal{L}_{\xor}$, then 
 \begin{align*}
 \exists \hat{\beta}_{\xor} &: z_{\xor} = c_{\xor} \xorsum{i=1}{\gamma_{\xor}}  g_{\xor,i} \hat{\beta}_{\xor,i} \, .
 \end{align*}
  Partitioning $\hat{\beta}_{\xor,1:\gamma_{\xor}}=\begin{bmatrix}\hat{\beta}_{1,1:\gamma_{1}}\,,\, \hat{\beta}_{2,1:\gamma_{2}}\end{bmatrix}$, it follows that there exist $z_1 \in \mathcal{L}_1$ and $z_2 \in \mathcal{L}_2$ such that $z_{\xor} = z_1 \xor z_2$. Therefore, $z_{\xor} \in \mathcal{L}_1 \xor \mathcal{L}_2$ and $ \mathcal{L}_{\xor} \subseteq \mathcal{L}_1  \xor \mathcal{L}_2$.
\end{proof}  
  
%
\subsubsection{Minkowski NOT ($\neg$), and XNOR ($\xnor$)} 

Given that we are able to do Minkowski XOR operation between logical zonotopes, we will able to find the Minkowski NOT with the following:

\begin{corollary}
\label{col:not}
Given a logical zonotope $\mathcal{L}=\langle c,G \rangle$, the Minkowski NOT can be computed exactly as: 
\begin{align}
     \neg \mathcal{L} &= \langle c \xor 1 ,G \rangle.
     \label{eq:negmink}
\end{align}
\end{corollary}

\begin{proof}
The proof follows directly from truth table of XOR gate and $\neg \mathcal{L}=\mathcal{L} \xor 1= \{z \xor 1 | z\in \mathcal{L} \} $ which results in inverting each binary vector in $\mathcal{L}$.
\end{proof}

Similarly, we can perform the Minkowski XNOR exactly as follows.
\begin{align}
\mathcal{L}_1 \xnor  \mathcal{L}_2 	& =  \neg( \mathcal{L}_1  \xor \mathcal{L}_2 ). 
\end{align}

\subsubsection{Minkowski AND} \,
Next, we over-approximate the Minkowski AND between two logical zonotopes as follows.
\begin{lemma}
\label{lem:and}
Given two logical zonotopes $\mathcal{L}_1=\langle c_{1},G_{1} \rangle$ and $\mathcal{L}_2 = \langle c_{2},G_{2} \rangle$, the Minkowski AND can be over-approximated by $\mathcal{L}_{\land}=\langle c_{\land},G_{\land} \rangle$: 
\begin{align}
      \mathcal{L}_1  \mathcal{L}_2  &\subseteq \mathcal{L}_{\land}\,. 
     \label{eq:andmink}
\end{align}
where
$c_{\land} {=} c_{1}  c_{2}$ and 
\begin{align}
G_{\land} {=}\big[ &c_{1} g_{2,1},\dots, c_{1} g_{2,\gamma_2}, c_{2} g_{1,1},\dots, c_{2} g_{1,\gamma_1},\nonumber\\ 
& g_{1,1}  g_{2,1}, g_{1,1}  g_{2,2}, \dots, g_{1,\gamma_1}  g_{2,\gamma_2}\big] \, .    
\end{align}
\end{lemma}
\begin{proof}
Choose $z_1 \in \mathcal{L}_1$ and $z_2 \in \mathcal{L}_2$. Then, we have
\begin{align}
\exists \hat{\beta}_{1} &: z_1 = c_{1} \xorsum{i=1}{\gamma_{1}}  g_{1,i} \hat{\beta}_{1,i}\, , \label{eq:zZ1} \\
\exists \hat{\beta}_{2} &: z_2 = c_{2} \xorsum{i=1}{\gamma_{2}}  g_{2,i} \hat{\beta}_{2,i}\, .\label{eq:zZ2} 
 \end{align}
AND-ing \eqref{eq:zZ1} and \eqref{eq:zZ2} results in
\begin{align}
z_1z_2=  & c_{1}  c_{2}  \xorsum{i=1}{\gamma_{2}} c_{1}  g_{2,i} \hat{\beta}_{2,i} \xorsum{i=1}{\gamma_{1}} c_{2} g_{1,i}  \hat{\beta}_{1,i} \nonumber \\ &\xorsum{i=1,j=1}{\gamma_{1}, \gamma_{2}}  g_{1,i}  g_{2,j} \hat{\beta}_{1,i} \hat{\beta}_{2,j}\, .
 \end{align}
 Combining the factors in 
 $$\hat{\beta}_{\land} {=} \begin{bmatrix} \hat{\beta}_{1,1:\gamma_{1}} ,\,  \hat{\beta}_{2,1:\gamma_{2}} ,\, \hat{\beta}_{1,1} \hat{\beta}_{2,1},\dots,\, \hat{\beta}_{1,\gamma_{1}} \hat{\beta}_{2,\gamma_{2}} \end{bmatrix}$$ 
 results in having $z_1z_2 \in \mathcal{L}_{\land}$ and thus $\mathcal{L}_1  \mathcal{L}_2 \subseteq \mathcal{L}_{\land}$.
\end{proof}

\begin{remark}
The term over-approximation in binary vectors with $\mathcal{L}_1 \subseteq \mathcal{L}_2$ means that $\mathcal{L}_2$ contains at least all the binary vectors contained in $\mathcal{L}_1$.
\end{remark}
\subsubsection{Minkowski NAND ($\nand$)} \,
Given that we are able to do Minkowski AND and NOT operations, we will be able to do the Minkowski NAND as follows.

\begin{corollary}
\label{col:nand}
Given two logical zonotopes $\mathcal{L}_1{=}\langle c_{1},G_{1} \rangle$ and $\mathcal{L}_2 {=} \langle c_{2},G_{2} \rangle$, the Minkowski NAND can be over-approximated by: 
\begin{align}
    \mathcal{L}_1 \nand \mathcal{L}_2 &= \neg (\mathcal{L}_1 \mathcal{L}_2).
     \label{eq:nandmink}
\end{align}
\end{corollary}
\begin{proof}
The proof follows directly from truth table of NAND function and Lemma  \ref{lem:and}.
\end{proof}

\subsubsection{Minkowski OR ($\vee$), and NOR ($\nor$)} \,
Given that we are able to NAND two sets which is a universal gate operation, we will be able to over-approximate the following logical Minkowski operations as shown next:
\begin{align}
\mathcal{L}_1 \lor  \mathcal{L}_2 	&= ( \neg \mathcal{L}_1 ) \nand ( \neg \mathcal{L}_2), \\
\mathcal{L}_1 \nor  \mathcal{L}_2 	&= \neg (\mathcal{L}_1 \lor \mathcal{L}_2).
\end{align}

\subsubsection{Computational Complexity}
For analyzing the computational complexity of the Minkowski logical operations, we have two logical zonotopes $\mathcal{L}_1=\langle c_{1},G_{1} \rangle$ and $\mathcal{L}_2=\langle c_{2},G_{2} \rangle$, where $c_1, c_2 \in \mathbb B^n$, $G_1 \in \mathbb B^{n\times\gamma_1}$ and $G_2 \in \mathbb B^{n\times\gamma_2}$. In Lemma~\ref{lem:xor}, we see that the Minkowski XOR only consists of $n$ binary operations for XORing the centers $c_1$ and $c_2$, resulting in a computational complexity of $\mathcal O(n)$. In other words, the complexity of Minkowski XOR scales linearly with the dimension of the binary vector space. Similarly, we see in Corollary~\ref{col:not} that applying the Minkowski NOT to a logical zonotope also has a computational complexity of $\mathcal O(n)$. By construction, XNOR also has complexity of $\mathcal O(n)$. In Lemma~\ref{lem:and}, we see that the Minkowski AND operation consists of ANDing the centers and generators with each other. Since each AND operation involves $n$ binary operations, the resulting computational complexity is $\mathcal O(n\gamma_1\gamma_2)$. Since the complexity of the Minkowski AND operation dominates the Minkowski NAND, OR, and NOR operations, they also have complexities of $\mathcal O(n\gamma_1\gamma_2)$. We list the operations and their corresponding complexities in Table~\ref{tab:complexity}.

\begin{table}[tbp]
\caption{Minkowski Logical Operation Complexity}
\label{tab:complexity}
\vspace{-2mm}
\centering
\normalsize
\resizebox{\columnwidth}{!}{
\begin{tabular}{l c c}
\toprule
 Operation & Complexity& Type \\
\midrule
XOR, NOT, XNOR & $\mathcal{O}(n)$ & Exact\\
AND, NAND, OR, NOR & $\mathcal{O}(n\gamma_1\gamma_2)$ & Over-approximation\\
\bottomrule
\end{tabular}
}
\vspace{-4mm}
\end{table}

\subsection{Minkowski Semi-Tensor Product}
Semi-tensor product has many application in different fields and is often useful in the analysis of logical systems \cite{conf:tensorproductsurvey}. In order to apply the Minkowski semi-tensor product to logical zonotopes, we first need to generalize logical zonotopes to logical matrix zonotopes. 

\begin{definition}(\textbf{Logical Matrix Zonotope}) \label{df:matzono}
Given a matrix $C  \in \mathbb{B}^{m\times n }$ and $\gamma {\in} \mathbb{N}$ generator matrices $G_i \in \mathbb{B}^{m\times n }$ in a generator list $\bar{G}{=} \big\{ G_1, \dots$  $,G_{\gamma} \big\}$, 
a logical matrix zonotope is defined as
\begin{align*}
\mathcal{L} = \Big\{ X \in \mathbb{B}^{m\times n } \; \Big| \; X = C \xorsum{i=1}{\gamma}  G_i \beta_i, \, \beta_i \in \{0,1\} \Big\} \, .
\end{align*}
We again use the shorthand notation $\mathcal{L} = \zono{C,G}$ for a logical matrix zonotope. 
\end{definition}

We define the  Minkowski semi-tensor product with a slight abuse of the notation. 
\begin{align}
    \mathcal{L}_1 \ltimes  \mathcal{L}_2 	&= \{z_1 \ltimes z_2| z_1\in \mathcal{L}_1, z_2 \in \mathcal{L}_2 \}\, . 
\end{align}

We compute the Minkowski semi-tensor product between two logical matrix zonotopes as follows.

\begin{lemma}
\label{lem:semitensorproductzonotopes}
Given two logical matrix zonotopes $\mathcal{L}_1=\langle C_{1},\bar{G}_{1} \rangle \in \mathbb{B}^{m \times n}$ and $\mathcal{L}_2 = \langle C_{2},\bar{G}_{2} \rangle \in \mathbb{B}^{p \times q}$, the Minkowski semi-tensor
product can be over-approximated by $\mathcal{L}_{\ltimes}=\big\langle C_{\ltimes} ,\bar{G}_{\ltimes} \big\rangle$: 
\begin{align}
    \mathcal{L}_1 \ltimes \mathcal{L}_2 &\subseteq \mathcal{L}_{\ltimes}, 
     \label{eq:nandmink}
\end{align}
where 
\begin{align}
    C_{\ltimes} {=}& C_{1} \ltimes C_{2},  \\
    \bar{G}_{\ltimes} {=}& \big\{ C_{1} {\ltimes} G_{2,1},\dots,C_{1} {\ltimes} G_{2,\gamma_1}, G_{1,1}  {\ltimes} C_{2} ,\dots,G_{1,\gamma_1}  {\ltimes} C_{2} \nonumber \\
    & \, \, G_{1,1}  {\ltimes} G_{2,1}  ,\dots,G_{1,\gamma_1}  {\ltimes} G_{2,\gamma_2}  \big\}.
\end{align}

\end{lemma}

\begin{proof}
 With $s$ as the least common multiple of $n$ and $p$, $s_1 =s/n$, and $s_2 =s/p$, choose $z \in \mathcal{L}_1 \ltimes \mathcal{L}_2$. Then, $\exists \hat{\beta}_{1},\hat{\beta}_{2}$ such that
\begin{align*}
z {=}& \Big(C_{1} \xorsum{i=1}{\gamma_{1}}  G_{1,i} \hat{\beta}_{1,i} \Big) \ltimes \Big(C_{2} \xorsum{i=1}{\gamma_{2}}  G_{2,i} \hat{\beta}_{2,i}\Big) \nonumber \\
{=}& \bigg( \Big(C_{1} \xorsum{i=1}{\gamma_{1}}  G_{1,i} \hat{\beta}_{1,i} \Big) \otimes I_{s_1} \bigg)  \bigg( \Big(C_{2} \xorsum{i=1}{\gamma_{2}}  G_{2,i} \hat{\beta}_{2,i} \Big) \otimes I_{s_2} \bigg) \nonumber \\
{=}& \Big(C_{1} {\otimes} I_{s_1} \xorsum{i=1}{\gamma_{1}}  G_{1,i} \hat{\beta}_{1,i} {\otimes} I_{s_1} \Big) \Big(C_{2} {\otimes} I_{s_2} \xorsum{i=1}{\gamma_{2}}  G_{2,i} \hat{\beta}_{2,i} {\otimes} I_{s_2} \Big) \nonumber \\
{=}& \Big(C_{1} {\otimes} I_{s_1}\Big)\Big(C_{2} {\otimes} I_{s_2\!}\Big) \xorsum{i=1}{\gamma_{2}} \Big(C_{1} {\otimes} I_{s_1}\!\Big) \Big(G_{2,i} \hat{\beta}_{2,i} {\otimes} I_{s_2}\! \Big) \nonumber \\
&\!\!{\xorsum{i=1}{\gamma_{1}}}  \Big(G_{1,i} \hat{\beta}_{1,i} {\otimes} I_{s_1} \Big) \Big(C_{2} {\otimes} I_{s_2}\!\Big)\nonumber \\
&\!\!\xorsum{i=1,j=1}{\gamma_{1},\gamma_{2}} \Big(G_{1,i} \hat{\beta}_{1,i} {\otimes} I_{s_1} \Big) \Big(G_{2,j} \hat{\beta}_{2,j} {\otimes} I_{s_2} \Big).
\label{eq:zZ1timeZ2} 
 \end{align*}
 Combining the factors in 
  $$\hat{\beta}_{\ltimes} {=} \begin{bmatrix} \hat{\beta}_{1,1:\gamma_{1}} ,\,  \hat{\beta}_{2,1:\gamma_{2}} ,\, \hat{\beta}_{1,1} \hat{\beta}_{2,1},\dots,\, \hat{\beta}_{1,\gamma_{1}} \hat{\beta}_{2,\gamma_{2}} \end{bmatrix}$$
 results in having $z \in \mathcal{L}_{\ltimes}$ and thus $\mathcal{L}_1 \ltimes \mathcal{L}_2 \subseteq \mathcal{L}_{\ltimes}$.
\end{proof}


%% file: Sections/Dreduce.tex
\subsection{Logical Zonotope Containment and Generators Reduction}

In certain scenarios, we might need to find a logical zonotope that contains at least the given binary vectors. One way to do that is as follows.
\begin{lemma}
\label{lm:enclosepoints}
Given a list $\mathcal{S}=\{s_1,\dots,s_p \}$ of $p$ binary vectors in $\mathbb{B}^{n}$, the logical zonotope $\mathcal{L}=\zono{c,G}$ with $s_{i} \in \mathcal{L}, \forall i = \{1,\dots,p\}$, is given by
\begin{align}
    c &= s_{1}, \\
    g_{i-1} &= s_{i} \xor c,\,\, \forall i = \{2,\dots,p\}.
\end{align}
\end{lemma}
\begin{proof}
By considering the truth table of all values of $\beta$, we can find that the evaluation of $\mathcal{L}$ results in  $c = s_{1}$ at one point and $g_{i-1}\xor c=s_{i} \xor c \xor c=s_{i} $, $\forall i = \{2,\dots,p\}$, at other points. 
\end{proof}

We propose Algorithm~\ref{alg:reduce} for reducing the number of generators while maintaining the same contained individual vectors. We first compute all the different binary vectors contained in the input logical zonotope $\mathcal{L}$ in Line~\ref{ln:evalfirst}. Then, the algorithm checks the effect of removing each generator by computing the binary vectors contained in the logical zonotope without the removed generator in Line~\ref{ln:evalsecond}. The chosen generator is deleted if its removal does not remove any binary vector in Lines \ref{ln:isequalcont} and \ref{ln:removegen}.



\begin{algorithm}[t]
\caption{Function \texttt{reduce} to reduce the number of generators of a logical zonotope.}
\label{alg:reduce}
\KwInput{A logical zonotope $\mathcal{L}=\zono{c,G }$ with large number $\gamma$ of generators } 
\KwOutput{A logical zonotope $\mathcal{L}_r=\zono{c_r,G_r }$ with $\gamma_r \le \gamma$ generators }
$c_r=c$ // Function \texttt{reduce} does not change the center \\
  $\mathcal{S}{=}$  \texttt{evaluate}($\mathcal{L}$) // Compute a list $\mathcal{S}$ of all binary vectors contained in $\mathcal{L}$ \label{ln:evalfirst}\\
  $G_r = G$ // Start with the same number of generators \\
  \For{$i = 1:\gamma$}{  
  $\mathcal{S}_r{=}$  \texttt{evaluate}($\mathcal{L}_r\, \backslash \, g_i$) // Compute a list $\mathcal{S}_r$ of all binary vectors contained in $\mathcal{L}_r$ without the generator $g_i$ \label{ln:evalsecond} \\ 
  \If{\texttt{isequal}($\mathcal{S},\mathcal{S}_r$)} 
  {\label{ln:isequalcont}
  $G_r=$ \texttt{removeGenerator}($G_r$,$g_i$) // Remove $g_i$ from $G_r$ \label{ln:removegen}
  }
  }
  $\mathcal{L}_r=\zono{c_r,G_r }$   
\end{algorithm}

%% file: Sections/Dreach.tex
\subsection{Reachability Analysis}
We aim to over-approximate the exact reachable region of \eqref{eq:feq} which is defined in Definition \ref{def:exactreachF} as follows.

\begin{theorem}
\label{thm:reach}
Given a logical function $f: \mathbb{B}^{n_x} \times \mathbb{B}^{n_u} \rightarrow \mathbb{B}^{n_x}$ in \eqref{eq:feq} and a set of possible inputs $\,\mathcal{U}_k \subset \mathbb{B}^{n_u}$ and starting from initial set $\hat{\mathcal{R}}_0 \subset 
\mathbb B^{n_x}$ where $x(0) \in \hat{\mathcal{R}}_0$, then the reachable region computed as
\begin{align}
    \hat{\mathcal{R}}_{k+1} =  f \big(\hat{\mathcal{R}}_{k},\mathcal{U}_k\big)
\end{align}
using logical zonotopes operations over-approximates the exact reachable set, i.e., $\hat{\mathcal{R}}_{k+1} {\supseteq} \mathcal{R}_{k+1}$.
\end{theorem}
\begin{proof} The logical function consists in general of XOR and NOT operations and any logical operations constructed from the NAND. $\forall x(k) \in \mathcal{R}_{k}$ and $u(k) \in \mathcal{U}_k$, we are able to compute Minkowski XOR and NOT exactly using Lemma \ref{lem:xor} and Corollary \ref{col:not} and over-approximate Minkowski NAND using Lemma \ref{lem:and} and Corollary \ref{col:nand}. Thus, $\hat{\mathcal{R}}_{k+1} \supseteq \mathcal{R}_{k+1}$.
\end{proof}

In Algorithm~\ref{alg:reach}, we overview an algorithm based on Theorem~\ref{thm:reach} for $N$-step reachability analysis using logical zonotopes. First, in Line~\ref{ln:contain}, we use Lemma~\ref{lm:enclosepoints} to convert the initial set of points ${\mathcal S}_0$ to get an initial logical zonotope $\bar{\mathcal{R}}_{0}$ which is further reduced to $\hat{\mathcal{R}}_{0}$ using Algorithm~\ref{alg:reach}. Then, we iterate $N$ times to find the $N$th-step reachable set as a logical zonotope.

\begin{algorithm}[t]
\caption{Reachability analysis for N-steps}
\label{alg:reach}
\KwInput{A logical function $f$, an initial set of points $\mathcal{S}_0$, a set of control input points $\mathcal{S}_{u,k}$, $\forall k=1,\dots,N$} 
\KwOutput{A reachable logical zonotope $\hat{\mathcal{R}}_{N}$ at the N-th step}
  $\bar{\mathcal{R}}_{0} =$ \texttt{enclosePoints}($\mathcal{S}_0$) // Enclose the set of points with a logical zonotope using Lemma~\ref{lm:enclosepoints}\label{ln:contain} \\
  $\hat{\mathcal{R}}_{0} =$ \texttt{reduce}($\bar{\mathcal{R}}_{0}$) // Reduce the number of generators using Algorithm~\ref{alg:reduce} \\
$\bar{\mathcal{U}}_{k} =$ \texttt{enclosePoints}($\mathcal{S}_{u,k}$), $\forall k=0,\dots,N-1$ \label{ln:containU} \\
${\mathcal{U}}_{k} =$ \texttt{reduce}($\bar{\mathcal{U}}_{k}$),  $\forall k=0,\dots,N-1$ \label{ln:reduceU}\\
  \For{$k = 0:N-1$}{  
    $\hat{\mathcal{R}}_{k+1} = f(\hat{\mathcal{R}}_{k}, \mathcal{U}_k)$ // Apply Minkowski logical operations
  }
\end{algorithm}

%% file: Sections/Eeval.tex
\section{Case Studies}\label{sec:eval}


To illustrate the use of operating over the generators in logical zonotopes, we present two different use cases. We first show how logical zonotopes can drastically improve the complexity of exhaustively searching for the key of an LFSR. Then, we formulate an intersection crossing problem, where we compare the computational complexity of BDDs, BCN-based semi-tensor products, and logical zonotopes when verifying the safety of four vehicles' intersection crossing protocol. 
The experiments were done on a processor 11$^{th}$ Generation Intel(R) Core(TM) i7-1185G7 with 16.0 GB RAM. 


\subsection{Exhaustive Search for the Key of an LFSR}
LFSRs are used intensively in many stream ciphers in order generate pseudo random longer keys from the input key. For simplicity we consider 60-bits LFSR $A$ initialized with the input key $K_A$ with length $l_k$. The operations on bit level are shown in Figure~\ref{fig:LSFR} where
\begin{align*}
    A[1] &= A[60] \xor A[59] \xor A[58] \xor A[14], \\
    \text{output} &= A[60] \xor A[59].
\end{align*}
Each bit $i$ of the output of the LFSR is XOR-ed with the message $m_A[i]$ to obtain one bit of the ciphertext $c_A[i]$. 

\begin{figure}
    \centering
    \includegraphics[scale=0.19]{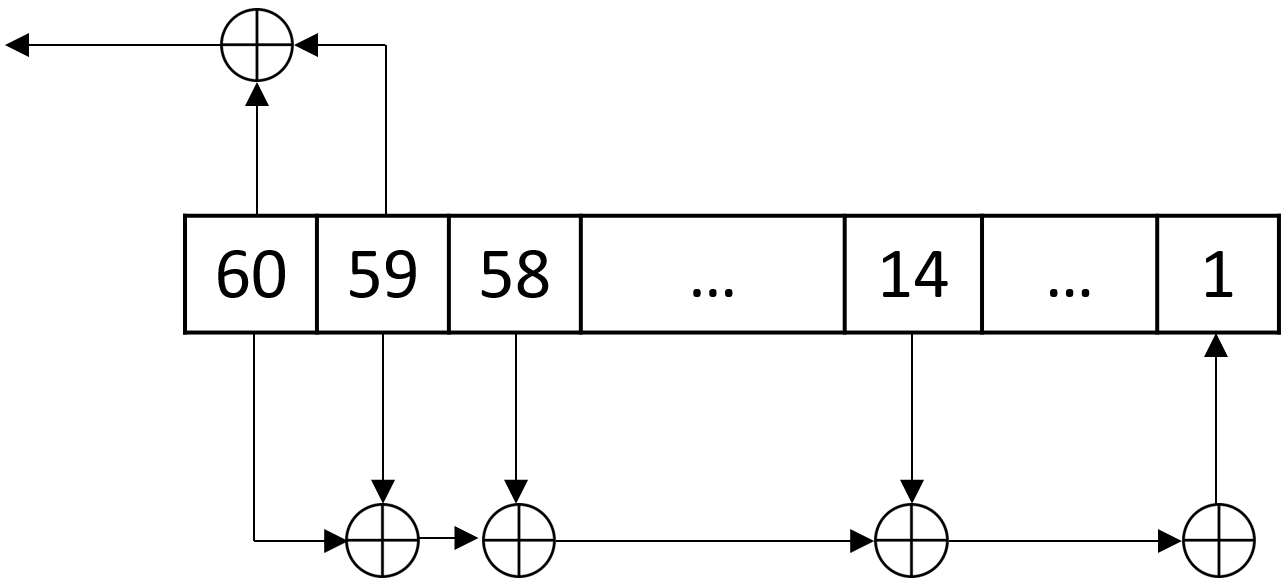}
    \caption{LFSR $A$.}
    \label{fig:LSFR}
    \vspace{-4mm}
\end{figure}

Now consider that we aim to obtain the input key $K_A$ using exhaustive search by trying out $2^{l_k}$ key values that can generate the cipher $c_A$ from $m_A$ with worst-case complexity $\mathcal{O}(2^{l_k})$ where $l_k=60$ is the key length. Instead, we propose to use logical zonotopes in Algorithm \ref{alg:LFSR} to decrease the complexity of the search algorithm. We start by defining a logical zonotope $\mathcal{L}_B$, which contains $0$ and $1$ in line \ref{ln:assignlog01}. Initially, we assign a logical zonotope to each bit of LFSR $A$ in line \ref{ln:bit2log} except the first two bits. Then, we set the first two bits of LFSR $A$ to one of the $2^2$ options of comb list in line \ref{ln:init2bits}. Then, we call the LFSR with the assigned key bits to get a list of logical zonotopes $\mathcal{G}_A$ with misuse of notations. The pseudo-random output of logical zonotopes $\mathcal{G}_A$ is XOR-ed with the message $m_A$ to get a list of ciphertext logical zonotopes $\mathcal{C}_A$. If any cipher of the list $c_A$ is not included in the corresponding logical zonotope $\mathcal{C}_A$, then the assigned two digits in line \ref{ln:init2bits} are wrong, and we do not need to continue finding values for the remaining bits of LFSR $A$. After finding the correct two bits with $c_A \in \mathcal{C}_A$, we continue by assigning a zero to bit by bit in line \ref{ln:keq0}. Then we generate the pseudo-random numbers $\mathcal{G}_A$ and XOR-ed it with the $m_A$ to get the list of cipher logical zonotopes $\mathcal{C}_A$. The cipher logical zonotopes $\mathcal{C}_A$ are checked to contain the list of ciphers $c_A$ and assign $\mathcal{K}_A$ in line \ref{ln:keq1}, accordingly. 
We measured the execution time of Algorithm \ref{alg:LFSR} with different key sizes in comparison to the execution time of traditional search in Table~\ref{tab:exectimekey}. To compute the execution time of the traditional search, we multiply the number of iterations by the average execution time of a single iteration. 


\begin{table}[tbp]
\caption{Execution Time (seconds) of exhaustive key search.}
\label{tab:exectimekey}
\vspace{-2mm}
\centering
\normalsize
\begin{tabular}{c  c c }
\toprule
 Key Size & Algorithm \ref{alg:LFSR} & Traditional Search \\
\midrule
30 & 1.97&  $1.18 \times 10^6$\\
60 &4.76&   $1.26 \times 10^{15}$\\
120 & 7.95& $1.46 \times 10^{33}$  \\
\bottomrule
\end{tabular}
\vspace{-4mm}
\end{table}

\begin{algorithm}[t]
\caption{Exhaustive search for LFSR key using logical zonotopes}
\label{alg:LFSR}
\KwInput{A sequence of messages $m_A$ and its ciphertexts $c_A$ with length $l_m$} 
\KwOutput{The used Key $\mathcal{K}_A$  with length $l_k$ in encrypting $m_A$}
$\mathcal{L}_B=$\texttt{enclosePoints}$([0\,\,1])$\, // enclose the points $0$ and $1$ by a logical zonotope \label{ln:assignlog01}\\
$\text{comb}=\{00,01,10,11\}$ \\
\For{$i = 3:l_k$}{
 $\mathcal{K}_A[i]=\mathcal{L}_B$ \label{ln:keqlogZono1} // assign the logical zonotope $\mathcal{L}_B$ to the key bits \label{ln:bit2log}\\
 }
\For{$i = 1:4$}{
  $\mathcal{K}_A[1:2]=\text{comb}[i]$  \label{ln:init2bits}\\
  $\mathcal{G}_A =$ \texttt{LFSR}($\mathcal{K}_A$) // generate pseudo random numbers from the key $\mathcal{K}_A$  \label{ln:lfsr1} \\
  $\mathcal{C}_A= \mathcal{G}_A \xor m_A$ \label{ln:enclog1} \\ 
  \If{$\neg$\texttt{contains}\,($\mathcal{C}_A$,$c_A$)} 
  {\label{ln:contains}
    \texttt{continue}; // continue if $c_A \notin \mathcal{C}_A$
  }
  \For{$j = 3:l_k$}{  
  $\mathcal{K}_A[j]=0$. \label{ln:keq0} \\ 
  $\mathcal{G}_A =$ \texttt{LFSR}($\mathcal{K}_A$) // generate pseudo random numbers from the key $\mathcal{K}_A$ \\
  $\mathcal{C}_A= \mathcal{G}_A \xor m_A$ \label{ln:enclog} \\ 
  \If{$\neg$\texttt{contains}\,($\mathcal{C}_A$,$c_A$)}
  {\label{ln:contains}
    $\mathcal{K}_A[i]=1$ // assign if $c_A \notin \mathcal{C}_A$ \label{ln:keq1}
  }
  }
  \If{\texttt{isequal}\,($\mathcal{K}_A \xor m_A$,$c_A$)}
  {\label{ln:isequal}
    \texttt{return} $\mathcal{K}_A$ \label{ln:return}
  }
  }
\end{algorithm}

\subsection{Safety Verification of an Intersection Crossing Protocol}

In this example, we consider an intersection where four vehicles need to pass through the intersection, while avoiding collision. For comparison, we encode their respective crossing protocols as logical functions and verify the safety of their protocols through reachability analysis using BDDs, a BCN semi-tensor product-based approach, and logical zonotopes. 
We denote whether vehicle $i$ is passing the intersection or not at time $k$ by $p_i(k)$. Then, we denote whether vehicle $i$ came first or not at time $k$ by $c_i(k)$. We use control inputs $u^p_i(k)$ and $u^c_i(k)$ to denote the decision of vehicle $i$ to pass or to come first at time $k$, respectively. 
For each vehicle $i=1,\dots,4$, the intersection passing protocol is represented by the following:
\begin{align}
    {p}_i(k+1) &= u^p_i(k)  \neg {p}_i(k)  \neg c_{i}(k).
    \label{eq:VehPassSafe} 
\end{align}
Then, the logic behind coming first for each vehicle $i=1,\dots,4$ is written as the following:
\begin{align}
    c_{i}(k+1) &= \neg p_{i}(k+1) ( u^c_{i}(k) \lor ( \neg p_{i}(k) p_{i}(k+1)) ).
    \label{eq:CameFirst}
\end{align}

To perform reachability analysis, we initialize the crossing problem with the following conditions: $p_1(0) = 1, \, p_2(0) \in \{0 , 1\}, \, p_3(0) = 0 , \, p_4(0) \in \{0 , 1\} , \, c_1(0) = 1,\, c_2(0) \in \{0 , 1\}, \, c_3(0) = 0 , \, c_4(0) \in \{0 , 1\}$.
To verify the passing protocol is always safe, under any decision made by each vehicle, we perform reachability analysis under the following uncertain control inputs: $u^p_1(k)\in \{0 , 1\},\,  u^p_2(k) =0,\,u^p_3(k)\in \{0 , 1\},\,  u^p_4(k) =0,\, u^c_1(k)\in \{0 , 1\}, \, u^c_2(k) \in \{0 , 1\},\,u^c_3(k)\in \{0 , 1\}, \, u^c_4(k) \in \{0 , 1\},\, k=0,\dots,N$.

\begin{table*}[tbp]
\caption{Execution Time (seconds) and number of points in each set (size) for verifying an intersection crossing protocol.}
\label{tab:exectimetrafic}
\centering
\normalsize
\begin{tabular}{c c c  c c  c c}
\toprule
 &  \multicolumn{2}{c}{Zonotope}  &\multicolumn{2}{c}{BDD} & \multicolumn{2}{c}{BCN}\\
  \cmidrule(lr){2-3} \cmidrule(lr){4-5} \cmidrule(lr){6-7}
 Steps $N$  & Time & Size &  Time & Size & Time & Size\\
\midrule
10 & 0.06& 16& 3.32& 14 & 7.75&  14\\
50 & 0.15&16 & 19.87 & 14& 48.40&14 \\
100 & 0.26&16 &39.78 & 14&104.91 &14 \\
1000 &1.84 &16 &406.60 & 14&1142.10 &14 \\
\bottomrule
\end{tabular}
\end{table*}



\begin{figure}
     \centering
     \begin{subfigure}[b]{0.14\textwidth}
         \centering
         \includegraphics[width=\textwidth]{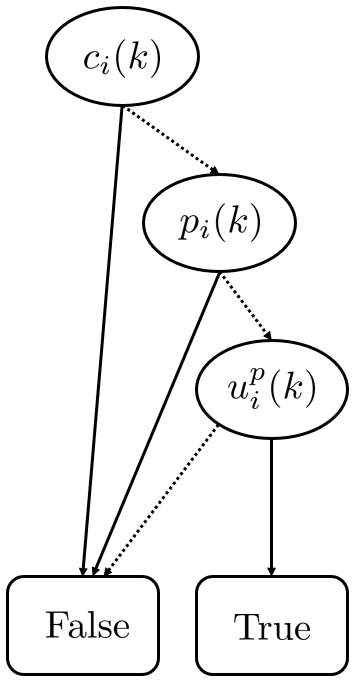}
         \caption{$p_i(k+1)$}
         \label{fig:BDD_VP}
     \end{subfigure}
     \begin{subfigure}[b]{0.165\textwidth}
         \centering
         \includegraphics[width=\textwidth]{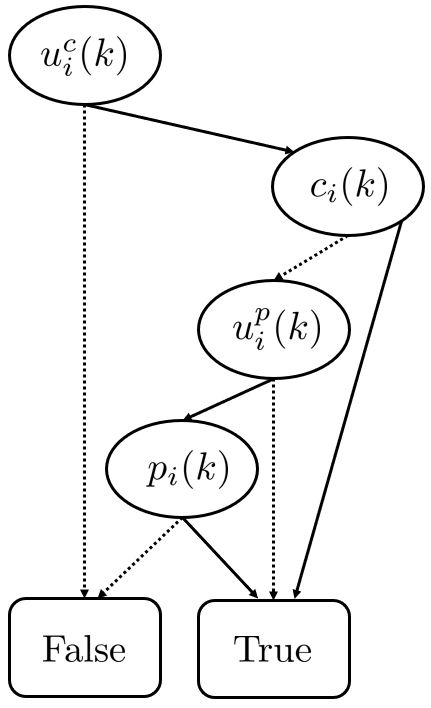}
         \caption{$c_i(k+1)$}
         \label{fig:Reduced_BDD_CF}
     \end{subfigure}
     \caption{Reduced BDDs for the intersection crossing example.}
     \label{fig:BDD}
     \vspace{-4mm}
\end{figure}

Then, we construct BDDs for each formula and execute the reduced form of the BDDs with uncertainty which is illustrated in Figure~\ref{fig:BDD}.
For the semi-tensor product-based approach with BCNs, we write state $x(k) {=} (\ltimes_{i=1}^4 p_{i}(k))$ $\ltimes (\ltimes_{i=1}^4 c_{i}(k))$. We write input $u(k) = (\ltimes_{i=1}^4 u^p_i(k)) \ltimes (\ltimes_{i=1}^4 u^c_i(k))$. The structure matrix $L$, which encodes \eqref{eq:VehPassSafe}-\eqref{eq:CameFirst}, is a $2^8 \times 2^{16}$ matrix where $8$ is the number of the states and $16$ is the number of states and inputs. We perform reachability analysis for the BCN using $x(k+1) = L \ltimes u(k) \ltimes x(k)$ for all possible combinations. For reachability analysis with logical zonotopes, we represent each uncertain variable in~\eqref{eq:VehPassSafe}-\eqref{eq:CameFirst} with a logical zonotope. We first compute the initial zonotope $\hat{\mathcal{R}}_0$ using Lemma \ref{lm:enclosepoints} which contains the initial and certain states.  
Then, using Theorem~\ref{thm:reach}, we compute the next reachable sets as logical zonotopes.
 
The execution time and the size of the reachable sets of the three approaches are presented in seconds in Table \ref{tab:exectimetrafic}. We note that reachability analysis using logical zonotopes provides better execution times when compared with reachability analysis with BDDs and semi-tensor products. Moreover, as the reachability analysis's time horizon increases, the reachability analysis's execution time with logical zonotopes increases slower than the other two methods. The logical zonotopes-based approach adds two extra points due to the over-approximation. 

%% file: Sections/Fcon.tex
\section{Conclusion}\label{sec:con}

This work proposes a novel set representation for binary vectors called logical zonotope. Logical zonotopes can represent up to $2^\gamma$ binary vectors using only $\gamma$ generators. We prove that applying different Minkowski logical operations to logical zonotopes always yields either exact solutions or over-approximations. 
In general, logical zonotopes allow for a variety of computationally efficient analyses of logical systems. In future work, we are investigating the potential of logical zonotopes for exploring the practical application of logical zonotopes in new use cases.





\balance

%% file: main.bbl
\begin{thebibliography}{10}
\providecommand{\url}[1]{#1}
\csname url@samestyle\endcsname
\providecommand{\newblock}{\relax}
\providecommand{\bibinfo}[2]{#2}
\providecommand{\BIBentrySTDinterwordspacing}{\spaceskip=0pt\relax}
\providecommand{\BIBentryALTinterwordstretchfactor}{4}
\providecommand{\BIBentryALTinterwordspacing}{\spaceskip=\fontdimen2\font plus
\BIBentryALTinterwordstretchfactor\fontdimen3\font minus
  \fontdimen4\font\relax}
\providecommand{\BIBforeignlanguage}[2]{{%
\expandafter\ifx\csname l@#1\endcsname\relax
\typeout{** WARNING: IEEEtran.bst: No hyphenation pattern has been}%
\typeout{** loaded for the language `#1'. Using the pattern for}%
\typeout{** the default language instead.}%
\else
\language=\csname l@#1\endcsname
\fi
#2}}
\providecommand{\BIBdecl}{\relax}
\BIBdecl

\bibitem{conf:BDDthesis}
A.~J. Hu, \emph{Techniques for efficient formal verification using binary
  decision diagrams}.\hskip 1em plus 0.5em minus 0.4em\relax stanford
  university, 1996.

\bibitem{conf:npcomplete}
G.~Cabodi, P.~Camurati, L.~Lavagno, and S.~Quer, ``Disjunctive partitioning and
  partial iterative squaring: An effective approach for symbolic traversal of
  large circuits,'' in \emph{Proceedings of the 34th annual Design Automation
  Conference}, 1997, pp. 728--733.

\bibitem{conf:effreachBDD}
M.~Byrod, B.~Lennartson, A.~Vahidi, and K.~Akesson, ``Efficient reachability
  analysis on modular discrete-event systems using binary decision diagrams,''
  in \emph{2006 8th International Workshop on Discrete Event Systems}.\hskip
  1em plus 0.5em minus 0.4em\relax IEEE, 2006, pp. 288--293.

\bibitem{7454743}
F.~Li and Y.~Tang, ``Robust reachability of boolean control networks,''
  \emph{IEEE/ACM Transactions on Computational Biology and Bioinformatics},
  vol.~14, no.~3, pp. 740--745, 2017.

\bibitem{leifeld2019overview}
T.~Leifeld, Z.~Zhang, and P.~Zhang, ``Overview and comparison of approaches
  towards an algebraic description of discrete event systems,'' \emph{Annual
  Reviews in Control}, vol.~48, pp. 80--88, 2019.

\bibitem{girard}
A.~Girard, ``Reachability of uncertain linear systems using zonotopes,'' in
  \emph{Hybrid Systems: Computation and Control}, M.~Morari and L.~Thiele,
  Eds.\hskip 1em plus 0.5em minus 0.4em\relax Berlin, Heidelberg: Springer
  Berlin Heidelberg, 2005, pp. 291--305.

\bibitem{conf:thesisalthoff}
M.~Althoff, ``{Reachability analysis and its application to the safety
  assessment of autonomous cars},'' Ph.D. dissertation, Technische
  Universit{\"{a}}t M{\"{u}}nchen, 2010.

\bibitem{conf:tensorproductsurvey}
D.~Cheng, H.~Qi, and A.~Xue, ``A survey on semi-tensor product of matrices,''
  \emph{Journal of Systems Science and Complexity}, vol.~20, no.~2, pp.
  304--322, 2007.

\bibitem{conf:semitensor}
H.~Qi and D.~Cheng, ``Analysis and control of boolean networks: A semi-tensor
  product approach,'' in \emph{7th Asian Control Conference}.\hskip 1em plus
  0.5em minus 0.4em\relax IEEE, 2009, pp. 1352--1356.

\bibitem{conf:zono1998}
W.~K{\"{u}}hn, ``Rigorously computed orbits of dynamical systems without the
  wrapping effect,'' \emph{Computing}, vol.~61, no.~1, pp. 47--67, 1998.

\bibitem{conf:constrainedzono}
J.~K. Scott, D.~M. Raimondo, G.~R. Marseglia, and R.~D. Braatz, ``Constrained
  zonotopes: A new tool for set-based estimation and fault detection,''
  \emph{Automatica}, vol.~69, pp. 126--136, 2016.

\bibitem{conf:hybridzono}
T.~J. Bird, H.~C. Pangborn, N.~Jain, and J.~P. Koeln, ``Hybrid zonotopes: A new
  set representation for reachability analysis of mixed logical dynamical
  systems,'' \emph{Automatica}, vol. 154, p. 111107, 2023.

\bibitem{Combastel2022}
C.~Combastel, ``Functional sets with typed symbols: Mixed zonotopes and
  polynotopes for hybrid nonlinear reachability and filtering,''
  \emph{Automatica}, vol. 143, p. 110457, 2022.

\end{thebibliography}
